\documentclass[letterpaper, 10 pt, conference]{IEEEtran}
\IEEEoverridecommandlockouts
\usepackage[left=1.91cm,right=1.91cm,top=2.54cm,bottom=1.91cm]{geometry}

\usepackage{cite}
\usepackage{amsmath,amssymb,amsfonts}
\usepackage[caption=false,font=footnotesize]{subfig}
\usepackage{graphicx}
\usepackage{textcomp}
\usepackage{xcolor}
\usepackage{siunitx}
\usepackage{url}
\usepackage{multirow}
\usepackage{multicol}
\usepackage{algorithm,algorithmic}
\newtheorem{theorem}{\bf Theorem}
\newtheorem{Definition}{\bf Definition}
\newtheorem{proof}{\bf Proof}

\begin{document}

\title{Controllability analysis of directed networks in finite states based on pruning motif isomorph}
\author{Jiarui Zhang, Jian Huang, Ji Guang and Jialong Gao
}

\maketitle

\begin{abstract}
	The current driver nodes search methods are difficult to cope with large networks, and the solution process does not consider the node cost. In order to solve the practical control problem of networks with different node costs in finite states, this paper proposes a pruning and motif isomorph search method for driver node set. Firstly, we prove the sufficient conditions for the network to be strictly controllable under partial nodes control, then we classify the nodes and prove the equivalence controllability of the pruning network, and then we establish three models of maximum augmenting path search, local pruning and motif matching to form a complete driver nodes set search algorithm. Finally, the algorithm is validated by real networks. The results show that our method not only guarantee the accuracy of search results, but also has the low time complexity, which can efficiently handle large networks, and no more than 16.84\% of the maximum driver nodes can control the whole network.
\end{abstract}


\section{Introduction}
\label{sec:introduction}
\IEEEPARstart{T}{he} current era is known as the "network era", and complex networks abstract various large and complex associations in the real world into network diagrams, which exist in all aspects of people's lives, including electric power networks\cite{b1}, biological networks\cite{b2}, financial networks\cite{b3}, transportation networks\cite{b4}, social networks\cite{b5} and so on. Complex networks exhibit rich dynamics and control problems have become an important research direction in network science. With the increasing scale of networks, how to effectively control these increasingly complex networks has become an important topic in control theory today\cite{b7}.

The control theory of complex networks originates from modern control theory\cite{b8,b9,b10,b11}, but has its own characteristics. Network control refers to the arrival of a network from an initial state to an arbitrary state in finite time under the action of appropriate inputs. In the real world, sometimes the network cannot reach the desired state by self-synchronization, in this case, it is necessary to add control to the network. In practical applications, it is not possible to control all nodes. To solve this problem, traction control has been widely used as a proven method.

Traction control\cite{b12,b13,b14} is the first control technique applied to complex networks, which is centered on solving two problems\cite{b15}: the feasibility of synchronizing the network states and the application of control signals to some nodes. The control method of implicating the whole network\cite{b16} not only requires a large amount of computation, but also has certain requirements on the network structure, which requires the network to have a large enough coupling strength. In order to solve the control problem of networks with different structures, Liu et al.\cite{b17} published a study on the structural controllability of complex networks in Nature in 2011, but the theory is only applicable to the directed complex networks with independently selectable edge weights. Yuan et al.\cite{b18} then proposed the theory of strict controllability, using the PBH criterion to solve the controllability of networks with deterministic edge weights and different structures. In addition, there are more studies on the controllability of different networks, including time-varying networks\cite{b19}, deeply coupled networks\cite{b20}, multilayer networks\cite{b21} and so on. However, there is a general problem that these methods are based on the controllability analysis of homogeneous complex networks, without considering the cost of imposing control on the nodes.

Many real-world complex system problems can be abstracted as network controllability problems. For example, in a power network, the overall power supply of a region is controlled by substations\cite{b22}; in a biological network, the selection of genetic nodes that act as drug targets makes the network of organisms develop well\cite{b23}; in a command and control network, the charge unit is fed with charge information to accomplish the mission, etc.\cite{b24}. In the above practical applications, it is impossible to control all the nodes in the network, and it is very difficult to find the nodes that control the network when the large scale complex network contains thousands of nodes. In addition, the cost of controlling different nodes is different and cannot be calculated directly by the method mentioned above.

To solve the above problem, this paper investigates the practical control problem of a directed complex network with different node costs in a finite state. By setting up a three-node motif, we give the sufficient conditions for the network to be strictly controllable under partial node control signals. In terms of network control theory, leaf nodes do not change the zero degree of the network\cite{b25}, which means that pruning the leaf nodes does not change the controllability of the network. Therefore, after classifying the nodes, we form a driver node search algorithm based on three models of maximum scalable path search, equivalent controllability of the network after pruning, and motif isomorphism matching. The experimental results show that our method can greatly reduce the computation of control inputs for large-scale networks and open up a new path for complex network control.

\subsection{Network Dynamics Model and Controllability Conditions}
Given a directed weighted complex network model as a binary group, $G = (V,W)$, where $V = \{ {v_i}|i = 1,2, \cdots ,N\}$ denotes the set of nodes in the network with total number $N$. Consider a continuous time dissipative coupled network consisting of $N$ nodes, based on the linearly coupled set of ordinary differential equations, the following state equation is established for the $i$th node in the network, as shown in \eqref{eq1}.
\begin{equation}
	\label{eq1}
	\dot{x}_i = f({x_i}) + \sum\limits_{j = 1}^N {{a_{ij}}} {x_j}
\end{equation}

where, ${x_i} = (x_{_i}^{(1)},x_{_i}^{(2)}, \cdots ,x_{_i}^{(n)}) \in {\Re ^n}$ is the state variable of node $i$, $f({x_i}) \in {\Re ^n}$ denotes the capability value of the node, and the coupling matrix $A = ({a_{ij}}) \in { \Re ^{N \times N}}$ satisfies the dissipative coupling condition $\sum {{a_{ij}}} = 0$.

The coupling matrix $A$ of the directed graph is not a symmetric matrix, so we need to redefine the coupling matrix $A$ to describe the topology of the directed weighted complex network. Based on the motif network\cite{b26}, the motif Laplacian matrix is constructed and defined as follows: firstly, 13 three-node motifs $M_k$ are defined, as shown in Fig. \ref{fig1}.

\begin{figure}[htp]
	\centerline{\includegraphics[width=0.9\columnwidth]{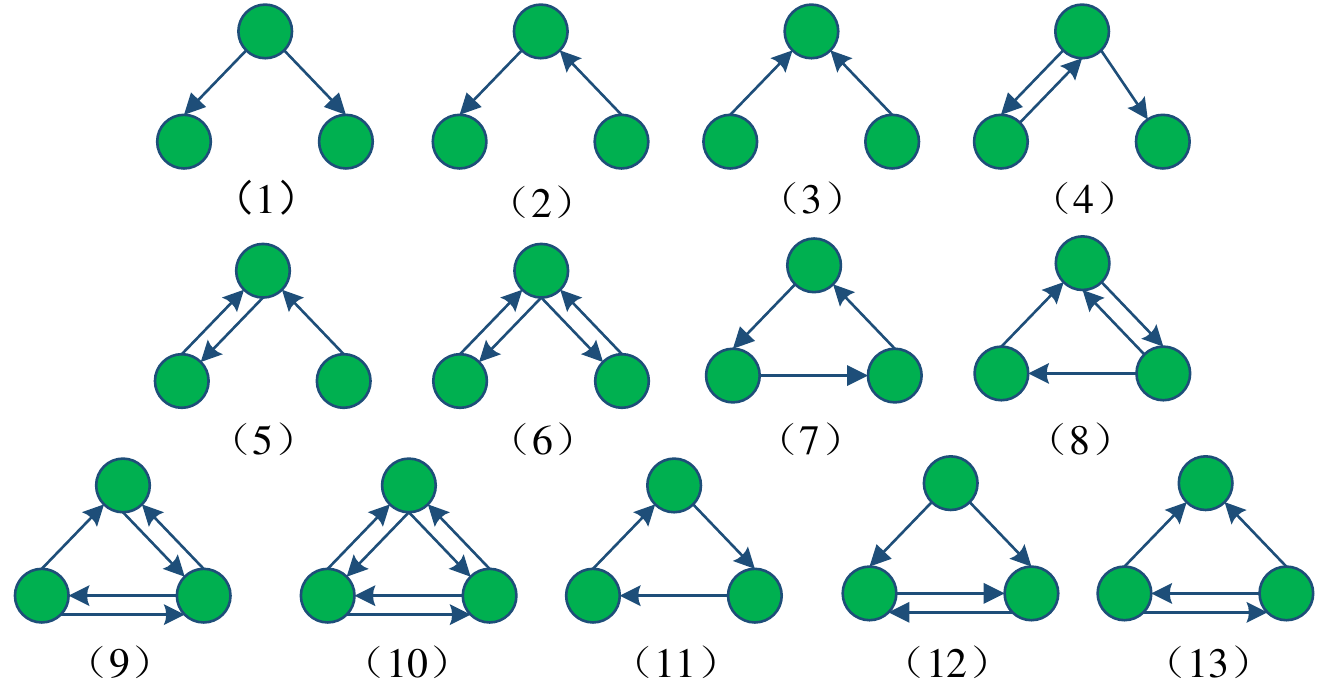}}
	\caption{Motifs of three-node.}
	\label{fig1}
\end{figure}

If there is a connection between node $i$ and node $j(i \ne j)$, and the connected edge appears ${\mu _{ij}}$ times in the module $M_k$, then ${a_{ij}} = \mu _{ij}$, otherwise, ${a_{ij}} = 0$. Second, the size of the diagonal elements is defined in this paper. Therefore, the matrix $A$ is shown in \eqref{eq2}.

\begin{equation}
	\label{eq2}
	\left\{
	{\begin{array}{*{20}{c}}
			{{a_{ij}} = {\mu _{ij}} \ge 0}\\
			{{a_{ii}} =  - \sum\limits_{j = 1.j \ne i}^N {{a_{ij}}}  =  - \sum\limits_{j = 1.j \ne i}^N {{a_{ji}}} }
		\end{array}}
		\right.
	\end{equation}
	
	The coupling matrix $A$ is our newly defined Laplacian matrix, which is used to represent the network structure. The matrix $A$ is an integrable real symmetric matrix with only one eigenvalue of 0, and the corresponding eigenvector is $(1,1,\cdots,1)^\mathrm{T}$, the rest of the eigenvalues are negative real numbers, and the eigenvectors form an $(N-1)$-dimensional subspace orthogonal to the eigenvectors $(1, 1,\cdots,1)^\mathrm{T}$. Diagonalize the matrix $A$ by ${A^\mathrm{T}} = B\Lambda {B^{ - 1}}$, where $B$ is the invertible square matrix, $\Lambda$ is the diagonal matrix, and $\Lambda = diag({\lambda _1},{\lambda _2}, \cdots ,{\lambda _N}) $.
	
	The model built by \eqref{eq1} can well describe the dynamic properties and topology of a directed or undirected complex network. When $A$ is a symmetric array, it can describe an undirected weighted complex network, and when it is an asymmetric array, it can also describe a directed weighted network. Therefore, for more generality, we use directed networks to illustrate the control to the steady state.
	
	\begin{Definition}
		There exist finite states ${s_i}(t), i = 1, 2, \cdots ,l$, when traction control is applied to some nodes in the network so that the whole network achieves synchronization in a certain state, that is, when there exists ${t_N} = t$, and the states of the nodes in the system ${x_1} = {x_2} = \cdots = {x_N} = s(t)$, the network achieves $s(t)$ state synchronization. By the dissipative coupling condition, the synchronized state $s(t) \in {\Re ^n}$ must be the solution of a single isolated node, satisfying $\dot{s}(t) = f(s(t)) = 0$.
	\end{Definition}
	
	The state equation of the network subjected to control signal input can be written as \eqref{eq3}:
	
	\begin{equation}
		\label{eq3}
		\left\{
		\begin{aligned}
			&\dot{x}_{ik} = f({x_{ik}}) + \sum\limits_{j = 1}^N {{a_{ij}}} {x_j} - {w_i}{d_i}({x_i}(t) - s(t))\hspace{-1cm}\\
			&\dot{x}_i = f({x_i}) + \sum\limits_{j = 1}^N {{a_{ij}}} {x_j}\\
			&i = 1,2, \cdots ,l\\
			&k = l + 1,l + 2, \cdots ,N
		\end{aligned}
		\right.
	\end{equation}
	
	Applying a control signal to the controlled node, ${d_i} \ge 0$ is the gain of node $i$ to reach the synchronized state, and $\omega_i$ is the cost of applying control to this node.
	
	Setting ${\eta _i}(t) = {x_i}(t) - s(t)$, where $i = 1,2, \cdots ,N$, we have \eqref{eq4}.
	\begin{equation}
		\label{eq4}
		\dot{\eta}_i(t) = f({x_i}(t)) - f(s(t)) + \sum\limits_{j = 1}^N {{a_{ij}}H{\eta _j}(t) - {d_i}H} {\eta _i}(t)
	\end{equation}
	
	By linear approximation of the synchronous state $s(t)$, we can get \eqref{eq5}.
	\begin{equation}
		\label{eq5}
		\begin{aligned}
			\dot{\eta} &= \eta [Jf(s)] + (A - D)\eta \\
			&= [Jf(s) + (A - D)]\eta
		\end{aligned}
	\end{equation}
	
	\noindent where, $Jf(s)$ is the Jacobi matrix of $f(x)$ in state $s(t)$, $\eta = {({\eta _1},{\eta _2}, \cdots ,{\eta _N})^\mathrm{T}}$, $D = diag({d_1},{d_2}, \cdots ,{d_N})$. When $i=1,2,\cdots, l$, ${d_i} > 0$. When $i =l+1, \cdots, N$, ${d_i} = 0$.
	
	Therefore, the local stability of the synchronous state $s(t)$ of the controlled network translates into the stability of the linear equation shown in \eqref{eq5}. When the feedback gain is $d \to \infty$, there exists a constant $\rho < 0$ such that $[Jf(s) + \rho {I_n}]$ is a $Hurwitz$ matrix, the network shown in \eqref{eq1} can be controlled to the synchronous state $s(t)$, as in \eqref{eq6}.
	\begin{equation}
		\label{eq6}
		\rho  \le {\lambda _N}(A - D)
	\end{equation}
	where ${\lambda _N}(A - D)$ is the maximum eigenvalue of the matrix $(A - D)$ and the minimum value is 0, so \eqref{eq6} can be satisfied. When the linear feedback gain is added to some nodes, that is, $\mathop {\lim }\limits_{d \to \infty } {\lambda _N}(A - D) = {\lambda _N}(\overline A )$, where ${\lambda _N}(\overline A )$ is the matrix obtained by removing the rows and columns in which control nodes $i_1, i_2, \cdots, i_l$ are located from the matrix $A$, and the synchronization state of the network is equivalent to \eqref{eq7}.
	\begin{equation}
		\label{eq7}
		\left\{
		\begin{aligned}
			&\dot{x}_i = s,i = 1,2, \cdots ,l\\
			&\dot{x}_i = f({x_i}) + \sum\limits_{j = 1}^N {{a_{ij}}} H({x_j}),i=l + 1,\\ &\hspace{0.9cm}l + 2, \cdots ,N
		\end{aligned}
		\right.
	\end{equation}
	
	On the other hand, it is necessary to ensure that some of the nodes of the network are in stable equilibrium after being controlled, so the following theorem is given.
	\begin{theorem}
		Suppose a dynamic system with state equation $X(t)=f[x(t), t]$, where $X(t)$ is the state vector, which is a function of state and time. If the system has a scalar function $V[X(t)]$  with positive definite and continuous first-order partial derivatives, in some neighborhood of the equilibrium state $X_e$, whose derivative $V$ is negative definite. When $\left\| X \right\| \to \infty$, there is $V \to \infty $, then the equilibrium state is consistently asymptotically stable over a large range.
	\end{theorem}
	
	\begin{proof}
		When $i = l + 1, l + 2, \cdots, N$, let $X(t)=FX(t)+AX(t)$, $F=diag\{{\gamma _1},{\gamma _2},\cdots ,where {\gamma _N}$ is the diagonal matrix, ${\gamma _i}(i = 1,2, \cdots ,N)$ are the node ${v_i}$ dynamics model parameters. Take the Lyapunov function as $V[X(t),t] = X^\mathrm{T}(t)PX(t)$, where $P$ is a positive definite matrix, with $V[X(t),t] > 0$, and $\left\| X \right\| \to \infty $, with $V[X(t),t] > 0$. The derivative for $V[X(t),t]$ can be obtained:
		\begin{equation}
			\begin{aligned}
				\frac{{dV[X(t),t]}}{{dt}} &= \dot{X}^{\mathrm{T}}(t)PX(t) + X(t)P \dot{X}^{\mathrm{T}}(t)\\
				&= {X^\mathrm{T}}(t)({F^\mathrm{T}} + F)X(t) + \\&\hspace{0.5cm} {X^\mathrm{T}}(t)({A^\mathrm{T}} + A)X(t)
			\end{aligned}
		\end{equation}
		
		\begin{equation}
			\begin{aligned}
				{X^\mathrm{T}}(t)({F^\mathrm{T}} + F)X(t) &= {X^\mathrm{T}}(t)diag\{ 2{\gamma _1}, 2{\gamma _2},\\
				&\hspace{0.5cm}  \cdots ,2{\gamma _N}\} X(t)\\
				&\le 2{\gamma _{\max }}{X^\mathrm{T}}(t)PX(t)
			\end{aligned}
		\end{equation}
		where ${\gamma _{\max }} = \max \{ {\gamma _1},{\gamma _2}, \cdots ,{\gamma _N}\}$, $A$ is the Laplacian matrix, ${A^\mathrm{T}} = T\Lambda {T^{ - 1}}$, $\Lambda = diag({\ lambda _1},{\lambda _2}, \cdots ,{\lambda _N})$ and $0 = {\lambda _1} \ge {\lambda _2} \ge \cdots \ge {\lambda _N}$.
		
		Similarly, ${X^\mathrm{T}}(t)({A^\mathrm{T}} + A)X(t) \le 2{\lambda _N}{X^\mathrm{T}}(t)PX(t)$, when $x(t) \ne s(t)$ has ${X^\mathrm{T}}(t) ({F^ \mathrm{T} } + F)X(t) > 0$ and ${X^\mathrm{T}}(t) ({A^\mathrm{T}} + A)X(t) \le 0$, therefore $\dot V^ [X(t),t] < 0$. According to Theorem 1, it is clear that the dynamical model constructed in this paper for the directed weighted network is consistently asymptotically stable over a large range and has a stable equilibrium state $s(t)$.
	\end{proof}
	
	In conclusion, the traction control of a part of nodes in a complex network can make the network with finite number of states reach a certain synchronous state as desired.
	
	\subsection{Search Model Based on Pruning Motif Isomorph}
	There are two sufficient conditions for the system to be controllable. One is from the paper of Liu et al.\cite{b17} on the complex network structures controllability published in Nature in 2011, which proposed to control the whole network by applying signals to the driver nodes. In addition, network controllability can be determined by Kalman criterion, and then the maximum matching theorem was proposed to compute the maximum driver node set by bipartite graph matching with the KM algorithm, but this method has high complexity and is difficult to perform fast computation for large-scale networks. Another one is the controllability judgment of undirected and directed, unweighted and weighted networks by PBH criterion by Yuan et al\cite{b18}. The maximum set of driver nodes is obtained by primary row transformation of the $[A - \lambda I]$ matrix, but this method also has a high complexity for large dense networks. One of the above two conditions holds that the network is controllable. Condition 2 makes up for the fact that condition 1 can only be used for directed networks, but neither of them takes into account the cost of applying control signals to the nodes.
	
	In order to solve the above problems, we construct a new search model. By analyzing the influence of different nodes on the controllability of the network, we classify the nodes, prune the network, reduce the computational effort, and finally build a complete search model for the driver node set by motif matching.
	
	\subsubsection{Node Classification}
	
	Applying control signals to different nodes in a complex network can have different effects. Taking a directed network with 5 nodes as an example, we can see that node 1 in Fig. \ref{fig2} always needs to apply a control signal, and it is optional whether nodes 3, 4 and 5 need to apply a control signal or not.
	
	\begin{figure}[htp]
		\centerline{\includegraphics[width=\columnwidth]{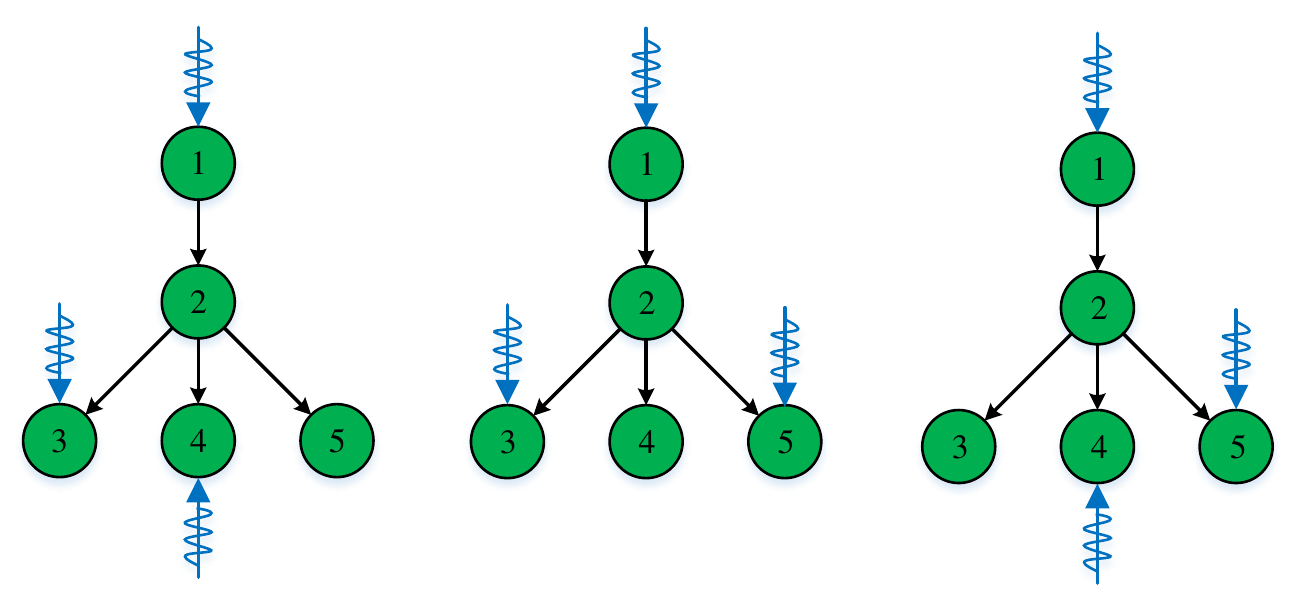}}
		\caption{Control signal application of five-node network.}
		\label{fig2}
	\end{figure}
	
	To address the shortcomings of the existing search methods for driver nodes, we first classify the nodes in the network, study the relationship between different nodes and the driver nodes, and perform the classification process to reduce the search space of the driver node set.
	
	In a directed network, nodes are connected to each other by directed edges, and considering the cost of nodes, changing the equation of state of the network only changes the value of the matrix $A$ on the diagonal, which does not affect the matrix controllability criterion. Therefore, the classification of nodes in this paper is consistent with the conclusion that the cost of nodes is zero, and nodes can be classified into the following four types by connecting edges:
	
	{\bf Type 1.} Isolated node, i.e., nodes with no input edge and no output edge, as shown in Fig. \ref{fig3} for node A, $\{ {v_{{N_i}}}\left| {{v_{{N_i}}} \in {V_N},i = 1,2 \cdots } \right.\}$, where ${V_N}$ is the set of isolated node.
	
	{\bf Type 2.} Incoming node, i.e., nodes with only input edges but no output edges, as shown as node B, $\{ {v_{{I_i}}} \left| {{v_{{I_i}}} \in {V_I},i = 1,2 \cdots } \right.\}$, where ${V_I}$ is the set of incoming nodes.
	
	{\bf Type 3.} Out-degree node, i.e., nodes with no input edges but only output edges, as shown as node C, $\{ {v_{{O_i}}}\left| {{v_{{O_i}}} \in {V_O},i = 1,2 \cdots } \right.\}$, where ${V_O}$ is the set of out-degree nodes.
	
	{\bf Type 4.} Fullness node, i.e., nodes with both input and output edges, are shown as node D, $\{ {{v_{{F_i}}}}\left| {{v_{{F_i}}} \in {V_F},i = 1,2 \cdots } \right.\}$, ${V_F}$ is the set of fullness nodes.
	
	\begin{figure}[htp]
		\centerline{\includegraphics[width=0.7\columnwidth]{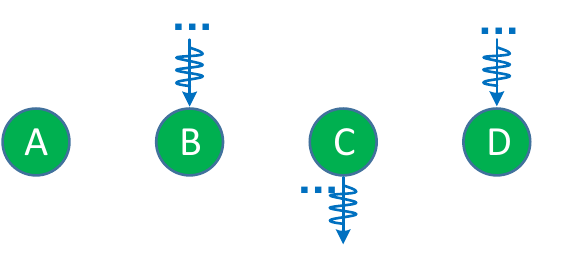}}
		\caption{Node classification.}
		\label{fig3}
	\end{figure}
	
	\subsubsection{Maximum Augmenting Path Search}
	
	\begin{Definition}
		The definition of augmenting path is that when a control signal is applied to a node, and the node is the source point, the control signal flows through the directed path to any node in the network, and the path that passes through it is called the augmenting path.
	\end{Definition}
	
	We can see that among the four types of nodes, two types of nodes, ${V_N}$ and ${V_O}$, do not have a directed edge pointed by the rest of the nodes in the network, i.e., these two types of nodes do not have a parent node and cannot pass control signals through other nodes in the network, so we need to apply control signals to them separately to make them as driver nodes. When the control signal is applied to ${V_O}$, the control signal is passed with the augmenting path $P$. Therefore, the longest directed path from this node can be a maximum augmenting path ${P_{\max }}$ of this driver node.
	
	After removing the maximum augmenting path in the network, if the control signal needs to pass through other nodes to the remaining isolated nodes, it must pass through the nodes on the removed maximum augmenting path, which is not consistent with the definition of network controllability that the control signal flows through a unique node. Therefore, the remaining isolated nodes need to be controlled separately. The specific process is shown in Algorithm 1.
	
	\begin{algorithm}[H]
		\caption{}
		\label{arg1}
		\begin{algorithmic}[1]
			\REQUIRE ~~\\
			Network, $G = (V,W)$\\
			Nodes cost, $cost(\cdot)$\\
			\ENSURE ~~\\
			Updated network, $G$\\
			Driver node set, $Driver$\\
			\STATE $Driver \leftarrow {V_N}$, $G \leftarrow (V - {V_N})$
			\IF{${V_I} \in V$}
			\STATE ${V_I} \leftarrow sort(cost({V_I}))$
			\FOR{$V_I$ in $sort({V_I})$}
			\STATE ${P_{\max }} \leftarrow DFS\& \min (cost(\sum\limits_{{x_i} \in P} {{x_i}} ))$
			\STATE $Driver \leftarrow \{ {y_i}\}$, $G \leftarrow (V - \{ {y_i}\} ),{y_i} \in {P_{\max }}$
			\STATE $Driver \leftarrow {V_N}$, $G \leftarrow (V - \{ {y_i}\} )$
			\ENDFOR
			\ENDIF
		\end{algorithmic}
	\end{algorithm}
	
	\subsubsection{Local Network Pruning}
	
	By pruning some of the augmenting paths in the network, we can reduce the network to pieces and greatly reduce the search space. The network contains nodes with minimum entry, including leaf structures with in-degree equal to 1, and we need to prune these nodes and their parent nodes until the entire network no longer contains. We prove the rationality and validity of this process with a single pruning.
	
	\begin{proof}
		The Laplacian matrix $A$ of the directed network is expressed in the following form. Also, for computational convenience, the node cost is assumed to be 0. Denoting the in-degree as the positive direction, the minimum in-degree node is assumed to be ordinal number 1, the length of vector $c$ can be expressed as the number of motif involved in this pruning node, the vector $\alpha$ and the vector $\beta$ denote the connected edges of this node and the parent node with the rest of the nodes, respectively. The matrix $A_0$ is a square matrix of order $(N - 3) \times (N - 3)$, which represents the connectivity of the remaining nodes after pruning this node and its parent nodes.
		
		\begin{gather*}A=
			\begin{pmatrix}
				0 & c & 0 & 0 \\
				c & 0 & \alpha^\mathrm{T} & 0\\
				0 & \alpha & 0 & \beta^\mathrm{T}\\
				0 & 0 & \beta & A_0
			\end{pmatrix}
			\quad
		\end{gather*}
		
		Assuming that the network is controllable, the maximum number of driver nodes is constant from the PBH criterion, ${N_D} = rank(\lambda {I_N} - A,B)$, $\lambda$ is any eigenvalue of matrix $A$, and $B$ is control input. According to the maximum matching theorem, the node has a parent node, which can not be a control node, therefore, the driver node set of the remaining network remains unchanged after pruning, i.e., ${N_D} = rank({\lambda _{{A_0}}}{I_{N - 3}} - {A_0},{B_0})$.
		
		In the following, we illustrate the controllability of the network in terms of the change of its rank. There exists an invertible matrix $D$ that makes \eqref{eq10} hold.
		
		\begin{equation}
			({\lambda _{{A_0}}}{I_{n - 3}} - {A_0},{B_0})D = ({I_{n - 3}},0)
			\label{eq10}
		\end{equation}
		where \begin{gather*}D=
			\begin{pmatrix}
				D_1 & D_2\\
				D_3 & D_4
			\end{pmatrix}
			B=\begin{pmatrix}
				0\\
				0\\
				B_0
			\end{pmatrix}
			\quad
		\end{gather*}
		
		At the same time, we set:
		\begin{gather*}S=
			\begin{pmatrix}
				1 & 0 & 0 & 0 \\
				0 & 1 & \alpha^\mathrm{T}D_1 & 0\\
				0 & 0 & 1 & \beta^\mathrm{T}D_3\\
				0 & 0 & 0 & 1
			\end{pmatrix}
			U=
			\begin{pmatrix}
				1 & 0 & 0\\
				0 & 1 & 0\\
				\alpha & D_{11}\beta & D
			\end{pmatrix}
		\end{gather*}
		
		\begin{gather*}rank(S(\lambda {I_N} - A,B)U) =\\
			\begin{pmatrix}
				\lambda & -c & 0 & 0 & 0 \\
				-c & \lambda-\alpha^\mathrm{T}D_1\alpha & \alpha^\mathrm{T}D_1\alpha-\alpha^\mathrm{T} & 0 & 0\\
				0 & -\alpha & \lambda-\beta^\mathrm{T}D_3\beta & 0 & -\beta^\mathrm{T}D_2\\
				0 & 0 & 0 & I_{N-3} & 0
			\end{pmatrix}
		\end{gather*}
		
		Now, when ${\beta ^T}{D_2} \ne 0$, $rank(S(\lambda {I_N} - A, B)\cdot U) = N$, both $S$ and $U$ are invertible matrices, and matrix multiplication does not affect the rank change, we have $rank(\lambda {I_N} - A, B) = N$, i.e., the network can be controlled by inputting control signals $B$ to the remaining nodes.
	\end{proof}
	
	From the above proof, we can see that the controllability of the remaining network is equivalent to the original network. Therefore, we can prune the nodes to achieve the effect of reducing computation and complexity. In the following, we give the pseudo code as shown in Algorithm 2.
	
	\begin{algorithm}[H] 
		\caption{}
		\label{arg2}
		\begin{algorithmic}[1]
			\REQUIRE ~~\\
			Updated network, $G = (V,W)$\\
			Nodes cost, $cost(\cdot)$\\
			\ENSURE ~~\\
			Pruning network, $G_{tem}$\\
			\STATE ${V_O} \leftarrow sort(cost({V_O}))$
			\FOR{${v_O}$ in ${V_O})$}
			\IF{$count(V_O) < 4$}
			\STATE ${G_{tem}} \leftarrow (V - {v_o})$
			\ELSE
			\STATE ${G_{tem}} \leftarrow (V - {v_o} - \min (\cos t\{ {v_{{o^ + }}}\} ))$, ${\rm{  }}s.t.{\rm{ }}{v_{{o^ + }}}$ $\rightarrow {v_o}$
			\ENDIF
			\ENDFOR
		\end{algorithmic}
	\end{algorithm}
	
	\subsubsection{Motif Isomorphism Matching}
	When there are only three nodes left in the network, it is also difficult to determine the driver node set. Therefore, we need to continue to classify the motifs, so that different motifs correspond to different control methods. In this paper, we classify 13 motifs into three categories, which are shown below.
	
	{\bf Class A motif} correspond to 4, 6 and 13 of the three-node motifs, respectively. The first two subplots in Fig. \ref{fig4} show the correct signal imposition method, and the third subplot shows the incorrect signal imposition method.
	
	According to the directed edge in the motif, when the control signal is applied to the red node in the figure, the control of the motif cannot be achieved by only one signal. When the control signal is applied to the blue node, any node can control the whole motif. Therefore, the less costly node among the correct nodes in the class A motif is chosen as the driver node.
	
	\begin{figure}[htp]
		\centerline{\includegraphics[width=\columnwidth]{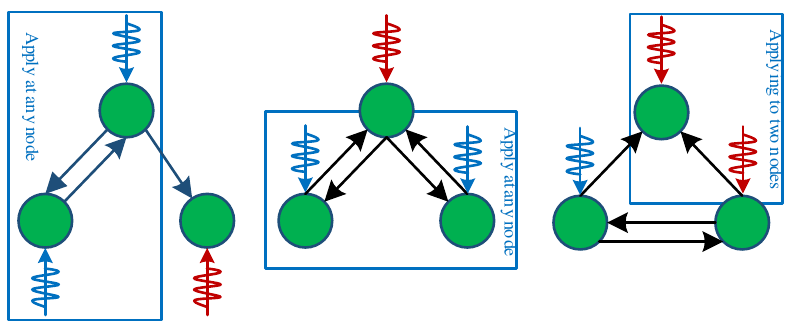}}
		\caption{Motif class A.}
		\label{fig4}
	\end{figure}
	
	{\bf Class B motif} correspond to 7 to 11 of the three-node motifs, respectively, and the first subfigure in Fig. \ref{fig5} shows the correct way to apply the signal. Each motif in the figure contains a directed loop, and the control signal applied to any node can be transmitted to the remaining two nodes, thus controlling the whole motif. Therefore, the less costly of the nodes in the class B motif is chosen as the driver node.
	
	\begin{figure}[htp]
		\centerline{\includegraphics[width=0.85\columnwidth]{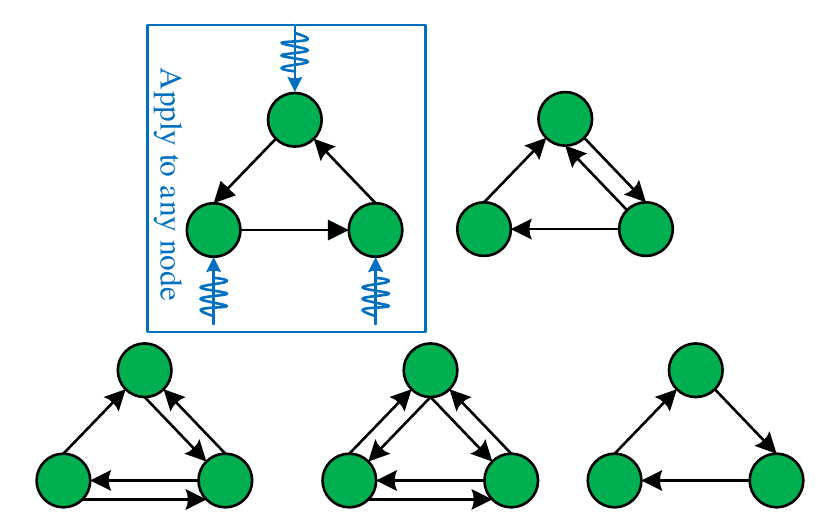}}
		\caption{Motif class B.}
		\label{fig5}
	\end{figure}
	
	{\bf class C motif} corresponds to 1 to 3, 5, 6 and 12 of the three-node motifs. Figure 6 shows the fixed control signal application positions, where the red signal is redundant and has no effect on the driver node. 3 motif requires two control signals, while the rest of the motifs require only one control signal. Class C motif can only apply the control signals shown in the figure, and the blue signal node is selected as the driver node.
	
	\begin{figure}[htp]
		\centerline{\includegraphics[width=0.9\columnwidth]{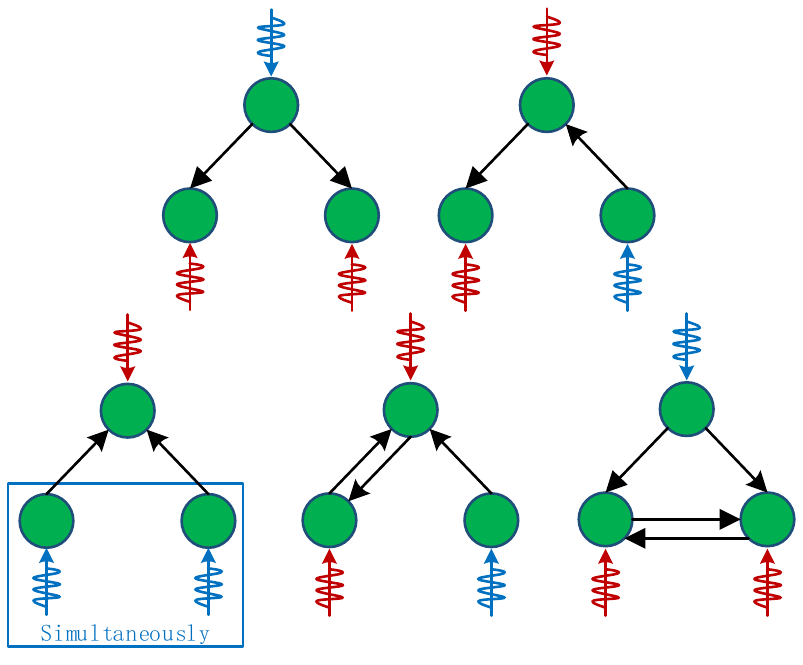}}
		\caption{Motif class C.}
		\label{fig6}
	\end{figure}
	
	Based on the above motif isomorphism matching method, we write the pseudo-code as shown in Algorithm 3.
	\begin{algorithm}[H] 
		\caption{}
		\begin{algorithmic}[1]
			\REQUIRE ~~\\
			Pruning network, ${G_{tem}}$\\
			Nodes cost, $cost(\cdot)$\\
			\ENSURE ~~\\
			Driver node, $Driver$\\
			\FOR{$motif$ in $match(G_{tem})$}
			\IF{$motif \in A$}
			\IF{$motif == motif(4)$}
			\STATE ${V_{de}} = V - {V_I}$
			\ELSE
			\STATE ${V_{de}} = V - {V_{in\_degree = 2}}$
			\STATE $Driver \leftarrow \min \cos t\{ {V_{de}}\}$
			\ENDIF
			\ELSIF {$motif \in B$}
			\STATE $Driver \leftarrow \min \cos t\{ V\}$
			\ELSE
			\STATE $Driver \leftarrow \min \cos t\{ {V_o}\}$
			\ENDIF
			\ENDFOR
		\end{algorithmic}
	\end{algorithm}
	
	\section{Algorithm Flow for Driver Node Set Search and Complexity Analysis}
	
	\subsection{Algorithm Flow}
	
	Based on the pruning motif isomorph search model, a complete driver node set search algorithm is formed. The three algorithms in the above model are sub-algorithms, which are implementations of the driver node set search algorithm under different conditions, and the general framework of the main algorithm is given below.
	
	\begin{algorithm}[H] 
		\twocolumn
		\caption{}
		\begin{algorithmic}[1]
			\REQUIRE ~~\\
			Network, ${G_{tem}}$\\
			Nodes cost, $cost(\cdot)$\\
			\ENSURE ~~\\
			Driver node set, ${N_D}$\\
			\IF{$V_O$ in $G$}
			\STATE $G,Driver \leftarrow$Algorithm\,1$(G)$
			\ELSE
			\WHILE{$number of V \in 2$}
			\STATE ${G_{tem}},Driver \leftarrow$Algorithm\,2$(G, cost)$
			\STATE $G,Driver \leftarrow$Algorithm\,3$(G_{tem}, cost)$
			\STATE $G,Driver \leftarrow$Algorithm\,1$(G_{tem}, cost)$
			\ENDWHILE
			\STATE ${N_D} \leftarrow Driver$
			\ENDIF
		\end{algorithmic}
	\end{algorithm}
	
	\subsection{Complexity Analysis}
	The method to find the minimum driver node set is usually reduced to a bipartite graph matching problem, and the classical algorithms are the Hopcroft-Karp algorithm with time complexity $O(V1/2*E)$, and the PBH criterion with time complexity $O(V3)$. The time complexity of our method is $O(lg(V-3)*(E+E/2+ . ...+E/root V))=O(log2(n+root n-3)*E)$. As we can see, as the number of nodes increases, especially after the number is greater than 20, the time complexity grows more slowly than other methods, and the time complexity is lower.
	
	\section{Experiment}
	Through the above reasoning and proof, the feasibility and credibility of our method are verified. We validate the effectiveness of our method on a real network by using four real networks as test datasets, searching for driver node set and analyzing the results. The size of nodes and edges varies from hundreds to hundreds of thousands, which are representative. Now, the experimental results are discussed and analyzed in the order of network size from small to large. 
	
	\subsubsection{Neural Network of Cryptobacterium histolytica}
	The analysis of neural networks can find the principles of organisms exhibiting the ability of learning, memory, exploration\cite{b28}, and complex movement. J. G. White et al\cite{b29} discovered that the neural network of Cryptobacterium hidradenum has more than 300 neurons and about 1000 cells in different connection groups. The neural network is a directed complex network in which different neurons have different stimulus costs, and the data were integrated and made publicly available by D. Watts and S. Strogatz\cite{b30}.
	
	\begin{figure}[htp]
		\centerline{\includegraphics[width=0.95\columnwidth]{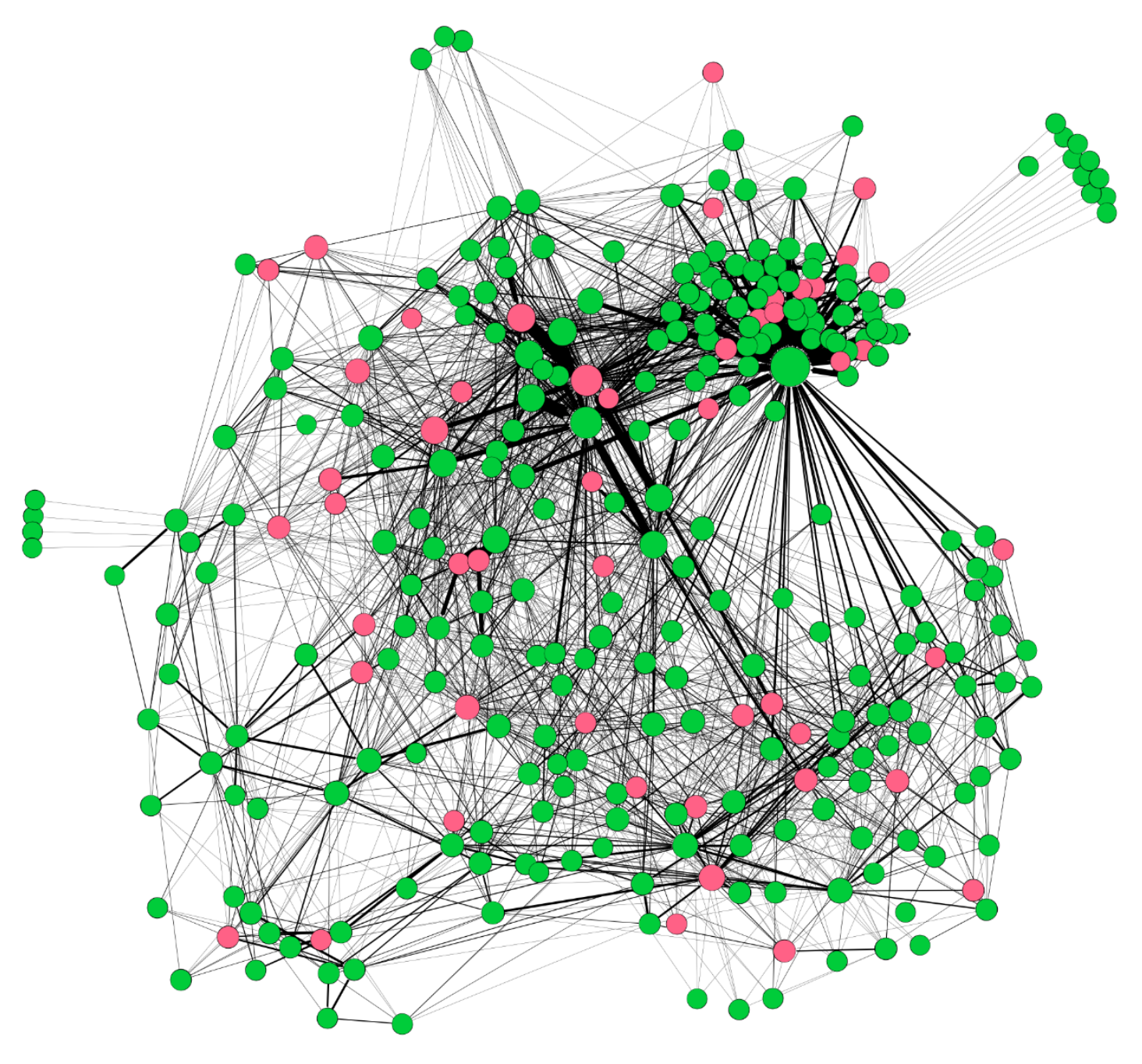}}
		\caption{Neural network of Caenorhabditis elegans and driver nodes.}
		\label{fig14}
	\end{figure}
	
	Fig. \ref{fig14} represents the directed complex network of the neural network, which contains 297 nodes and 2345 edges. The visual size of the nodes is positively correlated with the degree, and the greater the degree, the greater the display of the nodes. The nodes contain two colors, red nodes are the driver nodes, totaling 50, accounting for 16.84\% of the network, and green nodes are non-driver nodes, accounting for 83.16\%.
	
	\subsubsection{Blog Hyperlink Network}
	
	A hyperlinked directed network among blogs about American politics, mainly the hyperlinked information contained in blogs, was recorded by adam and Glance in 2005\cite{b31}.
	
	\begin{figure}[htp]
		\centerline{\includegraphics[width=\columnwidth]{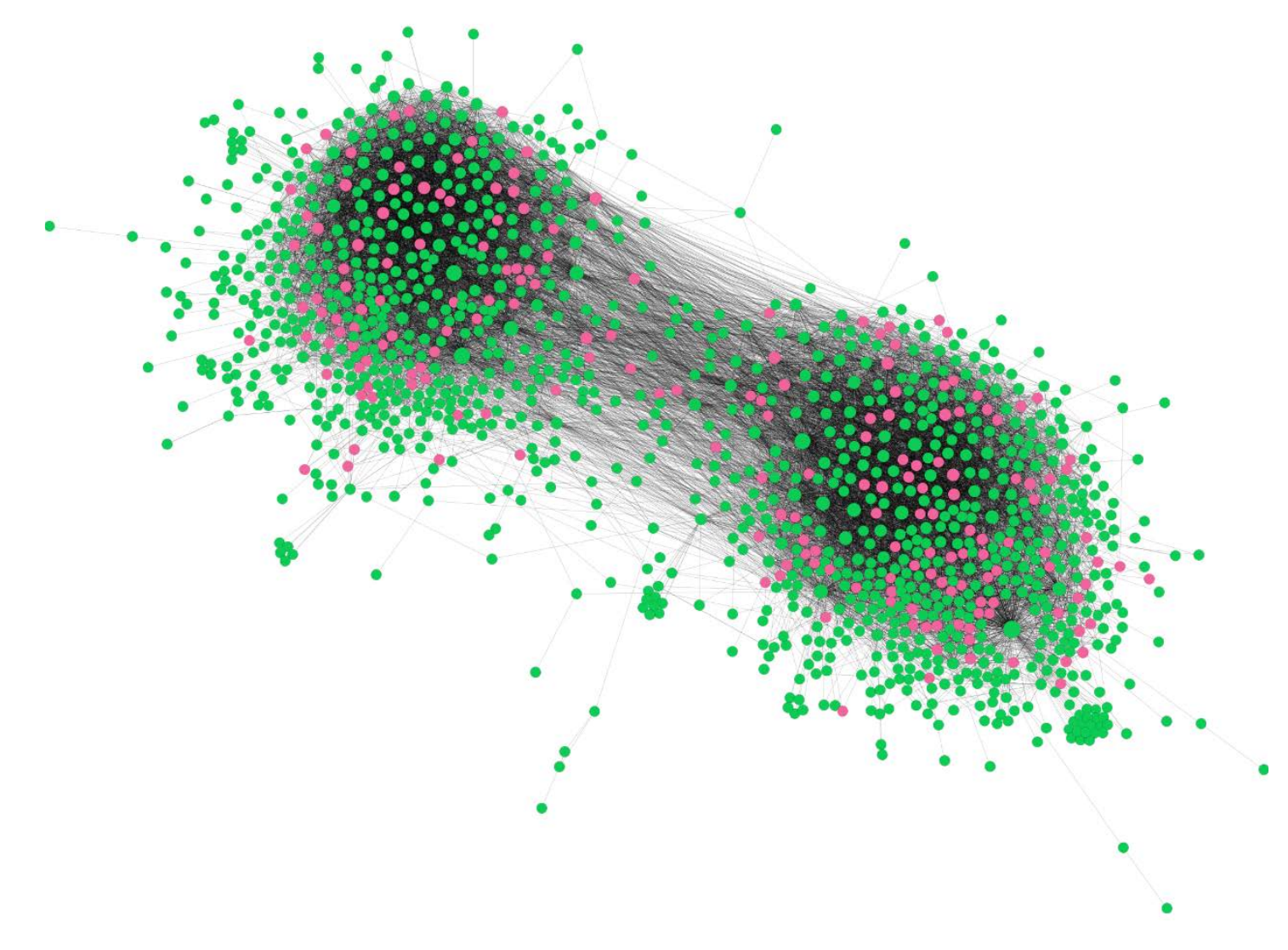}}
		\caption{Hyperlink network between blogs and driver nodes.}
		\label{fig15}
	\end{figure}
	
	Fig. \ref{fig15} represents the directed complex network of blog hyperlinks, containing 1224 nodes and 19022 edges, the node configuration in the network is consistent with the above figure, in which 192 nodes are driver nodes by red, accounting for 15.69\% of the overall nodes in the network.
	
	\begin{figure}[htp]
		\centerline{\includegraphics[width=\columnwidth]{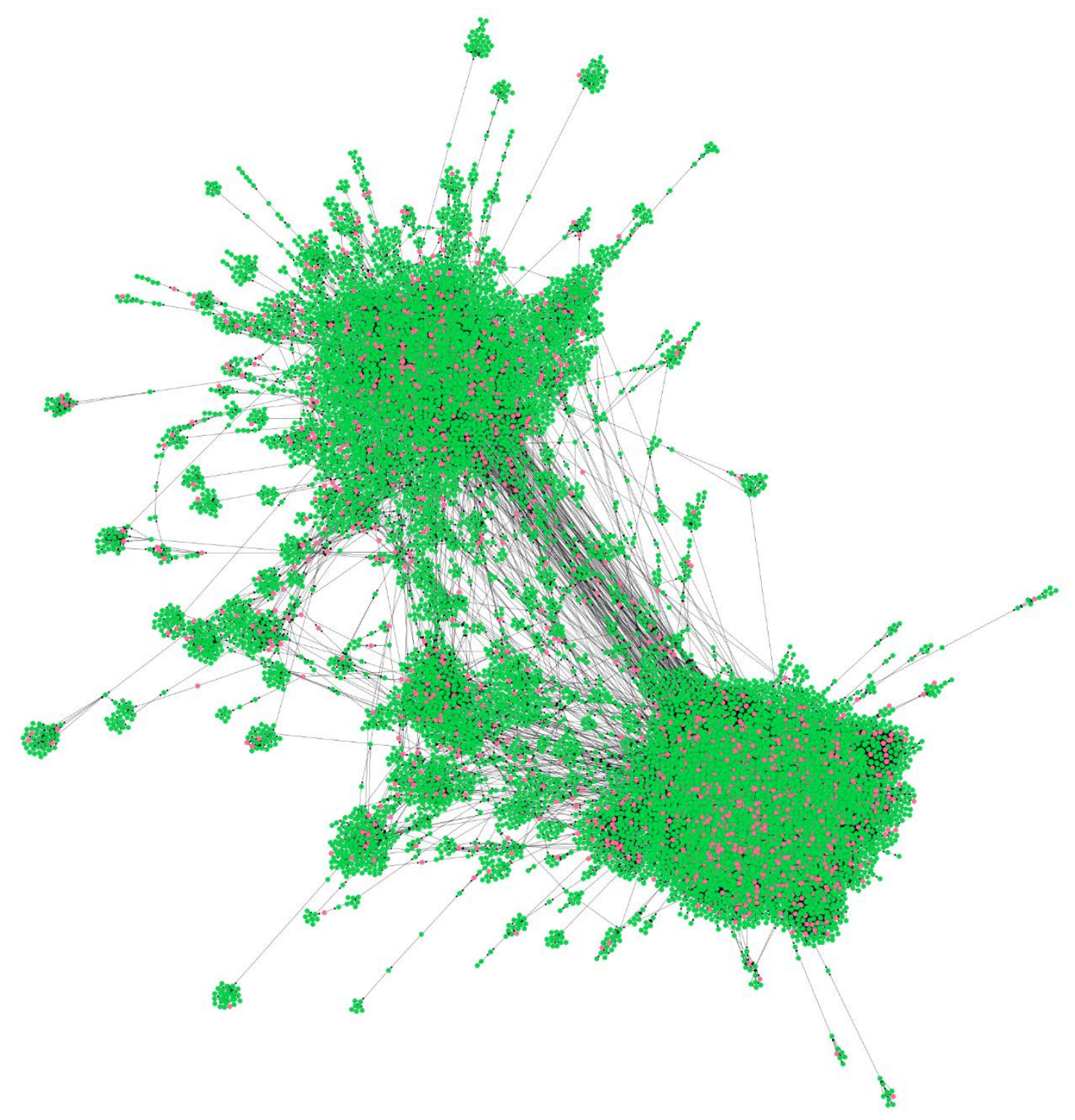}}
		\caption{Social networks for Deezer users and driver nodes.}
		\label{fig16}
	\end{figure}
	
	Figure 16 represents the directed complex network of small social attention relations, containing 28281 nodes and 92752 edges, the node configuration in the network is consistent with the above figure, among which 3566 nodes are driven by red nodes, accounting for 12.61\% of the network.
	
	\subsubsection{large social communication network}
	
	On July 4, 2012, CERN held a symposium and press conference to announce the discovery of a new particle with the elusive Higgs boson feature. The Higgs particle dataset focuses on the propagation data on Twitter before, during and after the launch of the new particle. Specifically, it contains all messages about the discovery from July 1 to 7, 2012, in which users used the "mention" action to form relationships with each other.
	
	\begin{figure}[htp]
		\centerline{\includegraphics[width=0.9\columnwidth]{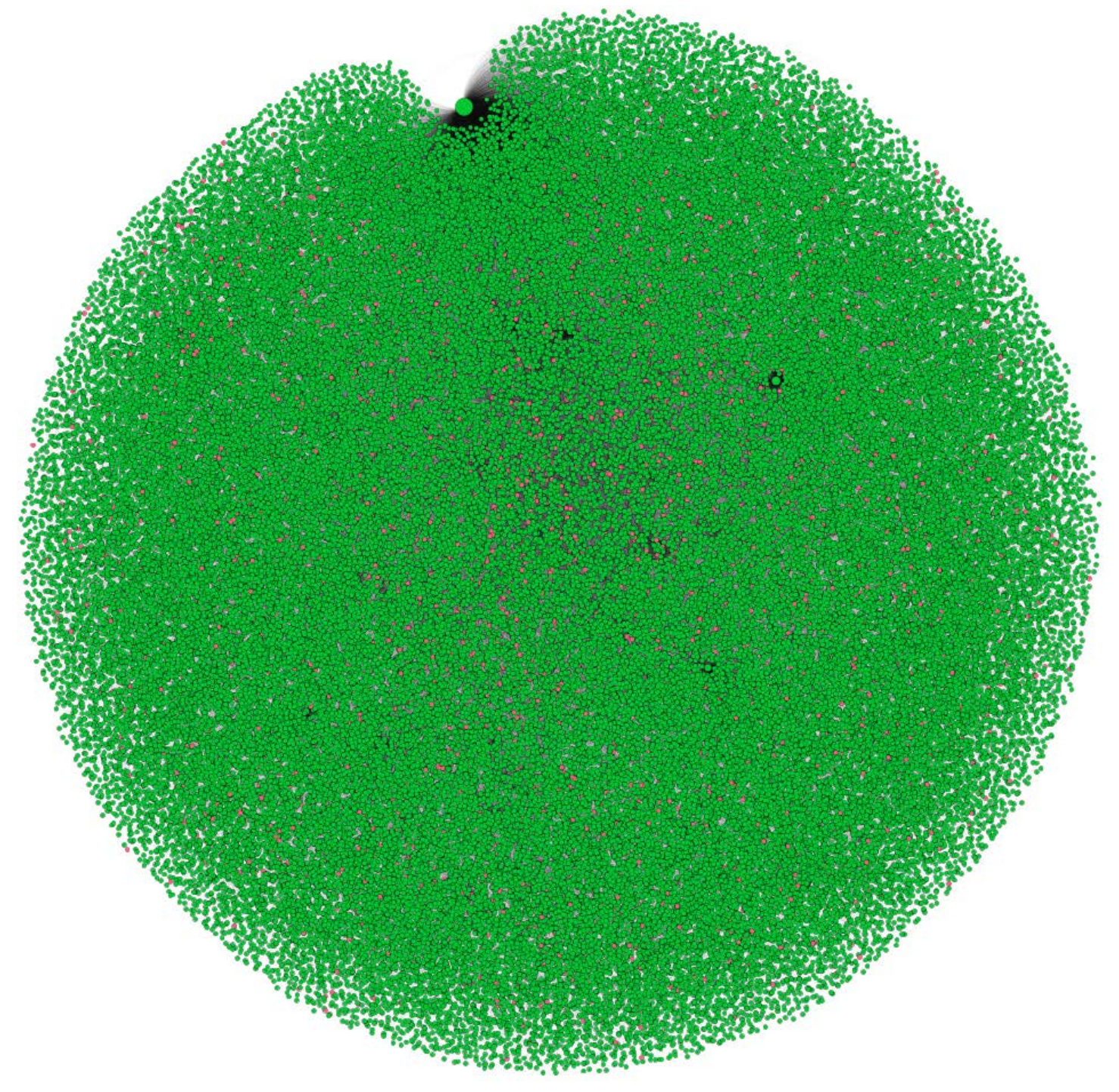}}
		\caption{Higgs message propagation network and driver nodes.}
		\label{fig17}
	\end{figure}
	
	Fig. 17 represents the directional complex network with 116408 nodes and 145774 edges for social network message propagation, the configuration of nodes in the network is consistent with the above figure, among which 3549 nodes are driver nodes by red, accounting for 3.05\% of the network.
	
	It can be seen that in the above-mentioned real networks, the proportion of driver nodes is not large, the maximum does not exceed 16.84\%, indicating that a small number of nodes can control the whole network, For a large and complex network with more than 300 nodes, no more than 1/6 of the nodes can be used as the driver node set to achieve control of the network. In addition, as the size of the network increases, the number of nodes and edges gradually increases, the proportion of the driver nodes decreases, which means that the proportion of the driver nodes is negatively related to the size of the network. It can be said that the more complex the network is, the richer the edges are, the smaller the proportion of the driver node is.
	
	At the same time, we also found another interesting point that the driver nodes are not the visually largest nodes, which are the nodes with the highest degree values, but their out-degree /in-degree ratios are higher.
	
	\section{conclusion}
	
	At present, the controllability of complex networks belongs to the cross research direction, which has attracted the attention of scholars in different fields. Network controllability theory has been applied to various real networks, but it is not enough to solve some practical problems of real networks, so we further improve the related theory and strengthen the wider application research.
	
	In this paper, we design a new method for searching the driver node set in complex networks without taking into account the node cost and the difficulty of efficiently coping with large networks, in order to solve the practical control problems of directed networks with different node costs in finite states. Firstly, we prove the sufficient conditions for the network to be strictly controllable under partial node control signal, and secondly, we develop a driver node set search algorithm based on the maximum augmenting path search, the equivalence controllability after network pruning, and the modal isomorphism matching model. Finally, the algorithm is validated by three types of experiments: examples, classical networks and real networks. The accuracy of the results is verified by illustrating the specific solving process of our method through examples. The efficiency of our method is verified by a large number of classical networks, and the connection between driver node set and the other nodes is found by the real network, which shows that a low percentage of driver nodes can control the whole network.


\begin{thebibliography}{00}
		
		\bibitem{b1} J. W. Wang and L. L. Rong, ``Cascade-based attack vulnerability on the US power grid,'' \emph{Safety Science}, vol. 47, no. 10, pp. 1332--1336, 2009.
		
		\bibitem{b2} E. Bullmore, O. Sporns, ``The economy of brain network organization,'' \emph{Nature Reviews Neuroscience}, vol. 13, pp. 336--349, 2012.
		
		\bibitem{b3} F. Schweitzer, Fagiolo G, Sornette D, et al. ``Economic Networks: The New Challenges,'' \emph{Science},vol. 325, no. 5935, pp. 422--425, 2009.
		
		\bibitem{b4} D. Li, B. Fu, Y. Wang, et al. ``Percolation transition in dynamical traffic network with evolving critical bottlenecks,'' \emph{Proceedings of the National Academy of Sciences}, vol. 112, no. 3, pp. 669--672, 2015.
		
		\bibitem{b5} S. P. Borgatti, A. Mehra, D. J. Brass, et al. ``Network Analysis in the Social Sciences. \emph{Science},'' vol. 323, no. 5916, pp. 892--895, 2009.
		
		\bibitem{b6} L. Qiao and X. C. Mao, ``Delay-induced complicated dynamics of a memristive Hopfield neural network,'' \emph{Journal of Dynamics and Control}, vol. 17, no. 4, pp. 384--390, 2019.
		
		\bibitem{b7} Y. Y. Liu and A. Barabási, ``Control principles of complex systems,'' \emph{Reviews of Modern Physics}, vol. 88, no. 3, pp. 035006, 2016.
		
		\bibitem{b8} X. Li, X. F. Wang and G. R. Chen, ``Pinning a Complex Dynamical Network to Its Equilibrium,'' \emph{IEEE Transactions on Circuits and Systems I: Regular Papers}, vol. 51, no. 10, pp. 2074–-2087, 2004.
		
		\bibitem{b9} X. F. Wang and G. R. Chen, ``Pinning Control of Scale-Free Dynamical Networks,'' \emph{Physica A: Statistical Mechanics and Its Applications}, vol. 310, no. 3–-4, pp. 521--531, 2002.
		
		\bibitem{b10} G. R. Chen, ``Problems and Challenges in Control Theory under Complex Dynamical Network Environments: Problems and Challenges in Control Theory under Complex Dynamical Network Environments,'' \emph{Acta Automatica Sinica}, vol. 39, no. 4, pp. 312--321, 2014.
		
		\bibitem{b11} A. Lombardi and M. Hrnquist, ``Controllability Analysis of Networks,'' \emph{Physical Review E}, vol. 75, no. 5, pp. 056110, 2007.
		
		\bibitem{b12} T. P. Chen, X. W. Liu and W. L. Lu, ``Pinning Complex Networks by a Single Controller,'' \emph{IEEE Transactions on Circuits and Systems I: Regular Papers}, vol. 54, no. 6, pp. 1317--1326, 2007.
		
		\bibitem{b13} W. L. Guo, F. Austin, S. Chen, et al. ``Pinning Synchronization of the Complex Networks with Non-Delayed and Delayed Coupling,'' \emph{Physics Letters A}, vol. 373, no. 17, pp. 1565--1572, 2009.
		
		\bibitem{b14} G. R. Chen and Z. S. Duan, ``Network Synchronizability Analysis: A Graph-Theoretic Approach,'' \emph{Chaos: An Interdisciplinary Journal of Nonlinear Science}, vol. 18, no. 3, pp. 037102, 2008.
		
		\bibitem{b15} X. F. Wang and H. S. Su, ``Recent Progress in Control of Complex Dynamical Networks,'' \emph{Advances in Mechanics}, vol. 38, no. 6, pp. 751--765, 2008.
		
		\bibitem{b16} J. Zhou, J. Lu and J. H. Lü, ``Pinning Adaptive Synchronization of a General Complex Dynamical Network,'' \emph{Automatica}, vol. 44, no. 4, pp. 996--1003, 2008.
		
		\bibitem{b17} Y. Y. Liu, J. Slotine and A. Barabási, ``Controllability of Complex Networks,'' \emph{Nature}, vol. 473, no. 7346, pp. 167--173, 2011.
		
		\bibitem{b18} Z. Z. Yuan, C. Zhao C and Z. Di, ``Exact Controllability of Complex Networks,'' \emph{Nature Communications}, vol. 4, no. 1, pp. 2447, 2013.
		
		\bibitem{b19} M. Pósfai and P. H\"{o}vel, ``Structural controllability of temporal networks,'' \emph{New Journal of Physics}, vol. 16, no. 12, pp. 123055, 2014.
		
		\bibitem{b20} J. N. Wu, X. Li and G. R. Chen, ``Controllability of Deep-Coupling Dynamical Networks,'' \emph{IEEE Transactions on Circuits and Systems I: Regular Papers}, vol. 67, no. 12, pp. 5211--5222, 2020.
		
		\bibitem{b21} Pósfai Márton, J. X. Gao, S. P. Cornelius, et al. ``Controllability of Multiplex Multi-Time-Scale Networks,'' \emph{Physical Review E}, vol. 94, no. 3, pp. 032316, 2016.
		
		\bibitem{b22} B. M. Zhang and S. S. Chen, \emph{Advanced Power Network Analysis}. Beijing, China: Tsinghua University Press, 2007.
		
		\bibitem{b23} S. Wuchty, ``Controllability in Protein Interaction Networks,'' \emph{Proceedings of the National Academy of Sciences}, vol. 111, no. 19, pp. 7156--7160, 2014.
		
		\bibitem{b24} L. L. Hou, S. Y. Lao, Y. D. Xiao, et al. ``Recent progress in controllability of complex network,'' \emph{Acta Physica Sinica}, vol. 64, no. 18, pp. 188901, 2015.
		
		\bibitem{b25} M. Bauer, O. Golinelli, ``Core Percolation in Random Graphs: A Critical Phenomena Analysis,'' \emph{The European Physical Journal B}, vol. 24, no. 3, pp. 339--352, 2001.
		
		\bibitem{b26} F. Monti, K. Otness and M. M. Bronstein, ``MotifNet: a motif-based Graph Convolutional Network for directed graphs,'' in Proceedings of the 2018 IEEE Data Science Workshop (DSW), Lausanne, Switzerland, 2018, pp. 225--228.
		
		\bibitem{b27} Y. Y. Wang, Z. Z. Yuan and C. Zhao, ``Analysis of Complex Networks Control Based on Core,'' \emph{Journal of Dynamics and Control}, vol. 19, no. 5, pp. 65--69, 2021.
		
		\bibitem{b28} J. R. Zhang, J. Huang, J. L. Gao and C. Zhou, ``Knowledge graph embedding by logical-default attention graph convolution neural network for link prediction,'' \emph{Information Sciences}, vol. 593, pp. 201--215, Feb, 2022.
		
		\bibitem{b29} J. G. White, E. Southgate, J. N. Thompson and S. Brenner, ``The Structure of the Nervous System of the Nematode Caenorhabditis Elegans,'' \emph{Philosophical Transactions of the Royal Society B: Biological Sciences}, vol. 314, pp. 305--340, 1986.
		
		\bibitem{b30} D. J. Watts, S. H. Strogatz, ``Collective Dynamics of small-World Networks,'' \emph{Nature}, vol. 393, no. 6684, pp. 440--442, 1998.
		
		\bibitem{b31} L. A. Adamic and N. Glance, ``The political blogosphere and the 2004 US Election: Divided They Blog,'' in Proceedings of the 3rd international Workshop on Link Discovery,'' Illinois, Chicago, USA, 2005, pp. 21--24.
		
		\bibitem{b32} B. Rozemberczki and R. Sarkar, ``Characteristic Functions on Graphs: Birds of a Feather, from Statistical Descriptors to Parametric Models,'' in Proceedings of the 29th ACM International Conference on Information \& Knowledge Management. New York, NY, USA, 2020, pp. 1325--1334.
		
	\end{thebibliography}
\end{document}